\documentclass[letterpaper,10pt,conference]{ieeeconf}
\IEEEoverridecommandlockouts
\usepackage{graphicx, subcaption}
\graphicspath{{.},{./Images/}}

\usepackage{algorithm2e}
\usepackage{booktabs}
\usepackage[usenames,dvipsnames]{xcolor}
\usepackage{etoolbox}

\usepackage{amsthm}

\renewcommand{\baselinestretch}{0.91}

\usepackage{color}
\usepackage{amsmath}
\usepackage{amssymb}
\usepackage{graphicx}
\usepackage{comment,xspace}
\usepackage{fancybox}

\newcommand{\real}{{\mathbb{R}}}
\newcommand{\reals}{\real}

\renewcommand{\natural}{{\mathbb{N}}}
\newcommand{\naturals}{\natural}

\newtheorem{theorem}{Theorem}[section]
\newtheorem{proposition}[theorem]{Proposition}
\newtheorem{lemma}[theorem]{Lemma}

\newtheorem{definition}[theorem]{Definition}

\newcommand{\argmin}{\operatornamewithlimits{argmin}}

\title{On the fundamental  limitations of performance for \\distributed decision-making in robotic networks}
\author{Federico Rossi \thanks{Federico Rossi and Marco Pavone are  with the Department of Aeronautics and Astronautics, Stanford University, Stanford, CA, 94305, \{{\tt\small frossi2}, {\tt\small pavone}\}{\tt \small @stanford.edu}} \and Marco Pavone}

\begin{document}
\maketitle
\thispagestyle{empty}
\pagestyle{empty}

\begin{abstract}
This paper studies fundamental limitations of performance for distributed decision-making in robotic networks. The class of decision-making problems we consider encompasses a number of prototypical problems such as average-based consensus as well as distributed optimization, leader election, majority voting, MAX, MIN, and logical formulas. We first propose a formal model for distributed computation on robotic networks that is based on the concept of I/O automata and is inspired by the Computer Science literature on distributed computing clusters. Then, we present  a number of bounds on time, message, and byte complexity, which we use to discuss the relative performance of a number of approaches for distributed decision-making. From a methodological standpoint, our work sheds light on the relation between the tools developed by the Computer Science and Controls communities on the topic of distributed algorithms.
\end{abstract}
\section{Introduction}
Decentralized decision-making in robotic networks is a ubiquitous problem, with applications as diverse as state estimation \cite{ROS:07}, formation control \cite{WR-RWB-EMA:07}, and cooperative task allocation \cite{MdW-BC:09}. In particular, the consensus problem, where the nodes in a robotic network have to agree on some common value, has received significant attention in the last decade following the works in \cite{AJ-JL-ASM:02, ROS-RMM:03c}. Most recent efforts in the Controls community have primarily focused on studying the properties and fundamental limitations of \emph{average-based} consensus, a subclass of the consensus problem where nodes average their state with their neighbors at each time step \cite{AO:10}. In these works, the dominant performance metric is time complexity, i.e., convergence time. In contrast, the computer science community has mainly focused its attention on the complementary notion of communication complexity, and ``communication-optimal" algorithms for selected consensus problems are now known \cite{NL:96}.

Despite the large interest in consensus problems in the last decade, little attention has been devoted to the problem of studying fundamental limitations of performance with respect to time and communication (in its broadest sense)  complexity. In \cite{AO:10} the author proposes a lower bound on the time complexity of \emph{average-based} consensus algorithms; several lower bounds on the \emph{message} complexity of specific instances of the consensus problem are known in the CS literature \cite{NL:96}, but no comprehensive study for more general distributed decision-making problems  is available. 

This motivates our work: in this paper, we explore tight lower bounds on the complexity of a large class of consensus problems with respect to metrics relevant to robotic systems. 

To appreciate the value of the time and communication complexity metrics, consider the following two scenarios.

Single-hop latency of robotic wireless communication protocols is often in excess of 10ms. The popular 802.15.4 protocol is bandwidth-limited: transmission of a 768 bit message containing the state of a 6DOF vehicle requires 7ms at the maximum allowable bit rate of  115200bps; latencies four to five times higher are typically observed in a laboratory environment. On the other hand, latency of high-bandwidth protocols such as 802.11 (WiFi) is greatly influenced by collisions: in presence of dozens of agents, latency is observed to consistently rise above 10ms.
Consider a typical network of 50 robotic agents in an arbitrary configuration. In Sec. \ref{sec:lbs} we note that the worst-case optimal time complexity of the consensus problem is $\Theta(n)$, with a constant factor close to one: thus, consensus can be achieved in approximately 500ms. On the other hand, the popular average-based consensus protocol has a convergence rate of $O(n^2) \cite{AO-JNT:07}$: convergence can require as much as $25s$ even on a static network. Intuitively, choosing a suitably fast consensus algorithm allows to track and control systems with orders-of-magnitude faster dynamic behavior.

On the other hand, energy consumption is also a concern in cyber-physical networks. Consider a swarm of autonomous underwater vehicles (AUV) tasked with performing a collaborative mission such as patrolling. Underwater ultrasonic communication has significantly higher energy demands than radio transmissions:  for instance, in \cite{MS-TS-TT:92}, the authors use 18W to maintain a 16 kbps link (sufficient to stream rich telemetry or low-quality images) with $35^\circ$ antennas over 6500 m. For underwater operations, omnidirectional communication is unadvisable because of intersymbol interference caused by multipath propagation, a phenomenon exacerbated by wide-beam transceivers. Hence, underwater communication typically relies on directional antennas and communication with $n$ agents requires $n$ \emph{different} messages. Consider, now, an all-to-all communication scheme to reach consensus (as needed by flooding and average-based consensus algorithms): in this case, communication with 50 AUVs would require 900 W per vehicle, which is impractical (as a comparison, electric motors on modern Remotely Operated Vehicles (ROVs) such as NOAA's Autonomous Benthic Explorer only draw 100 W in cruise \cite{DRY-AMB-BBW:91}). Hence, in this setting, a time-optimal algorithm such as flooding cannot be implemented. In contrast, a communication-optimal algorithm such as GHS  \cite{RGG-PAH-PMS:83} requires $O(\log_2 n)$ messages per AUV, resulting in a more practical power requirement of  102 W.

We mainly restrict our analysis to static networks, but also provide some extensions to time-varying network topologies. The contribution of our work is threefold: we extend results from the computer science literature to a broad class of distributed decision-making problems (collectively referred to as generalized consensus)  relevant to the control systems and robotic community; we also present a unified complexity theory for generalized consensus on static networks, identifying lower bounds on performance for metrics relevant to robotic systems. Finally, we discuss  algorithms that simultaneously make one or more of these bounds tight.

  The paper is structured as follows. In Section \ref{sec:background} we propose a formal model for robotic networks, complexity measures and a rigorous definition of the consensus problem that encompasses problems including (weighed) mean as well as MAX, MIN and voting. In Section \ref{sec:lbs} we present lower bounds on the time, message, and byte complexity of the consensus problems for \emph{sparse} and \emph{dense} networks. In Section \ref{sec:upper} we show tightness of these lower bounds under mild assumptions. Finally, in Section \ref{sec:conclusions} we draw our conclusions and we discuss directions for future work. Lower bounds and optimality results are summarized in Table \ref{tab:consbounds}.

We remark two rather interesting findings. On static networks, a modified version of the GHS minimum spanning tree algorithm solves the consensus problem with \emph{simultaneously} optimal time, message, byte and storage complexity. Intuitively, if ``few'' link reconfigurations are expected, algorithms based on construction of a spanning tree are superior to nonhierarchical schemes.
Conversely, on time-varying networks, a simple \emph{flooding} algorithm Pareto-dominates average-based consensus in all metrics we investigate except for storage complexity (arguably a minor concern in modern robotic systems). However, in practical implementations average-based consensus may be preferable to flooding whenever time performance is not critical and \emph{bandwidth} (whose role we discuss in \cite{FR-MP:14b}) is limited.

\begin{table*}[h!tb]{\footnotesize
\begin{subtable}[h]{\textwidth}
\centering
\begin{tabular}{rcccc}
& Time & Message (S) &Message (D) & Msg. (broadcast) \\
\toprule
Lower bound & n& $n\log n$ & $n^2$ & $n\log n$\\ 
\hline
Flooding (no failures) & \textbf{n} & $n^2 \log n$ & $n^3$ &$ n^2$\\ 
GHS modified (no failures) &$\textbf{n}$ \cite{BAw:87} & $\mathbf{n\log n}$  & $\mathbf{n^2}$ & $\mathbf{n\log n}$\\
Avg. based (no failures, $\varepsilon$) & $n^2\log (1/\varepsilon)$& $n^3 \log n$& $n^4$& $n^3$\\
Hybrid clustering \cite{FR-MP:13} &$n\log (n/m)$& $ 2mn + k|E_c|m $ & $ 2mn + k|E_c|m $ & $n(m+\log(n/m)$\\
\hline
Flooding (D-connectivity) &$\mathbf{nD}$ & $n^2D \log n $&$n^3D$& $n^2D$\\
Avg.-based (D-connectivity) &$n^2D\log(1/\varepsilon)$&$n^3D\log n$&$n^4D $&$n^3D$\\
\bottomrule
\\
& Byte (S)&Byte (D)& Byte (broadcast) &Storage\\
\toprule
Lower bound &  $(n\log n)(\log n + b)$  &  $n^2(\log n + b)$ &$(n\log n)(\log n + b) $& $\log n + b$\\ 
\hline
Flooding (no failures) & $ n^2\log n (\log n +b)$ & $ n^3(\log n +b)$ &$n^2(\log n + b) $ & $n(\log n+b)$\\ 
GHS modified (no failures) &$\mathbf{(n\log n)(\log n + b)}$& $\mathbf{n^2(\log n + b)}$&  $\mathbf{(n\log n)(\log n + b)}$ & $\mathbf{\log n + b}$\\
Avg. based (no failures, $\varepsilon$) &  $n^3 \log n(\log n + b) $&$n^4(\log n + b) $& $n^3 (\log n + b) $ & $\mathbf{\log n + b}$\\
Hybrid clustering \cite{FR-MP:13} & $m(n+k|E_c|)(\log n+b)$ &  $m(n+k|E_c|)(\log n+b)$ & $n(m+\log(n/m)\log n$ &$m(\log n + b)$ \\
\hline
Flooding (D-connectivity) &$n^3D\log  n(\log n+b)$&$n^4D (\log n+b)$ & $n^3D(\log n + b)$ &$n(\log n + b)$\\
Avg.-based (D-connectivity) &$n^3D\log n(\log n+b)$&$n^4D (\log n+b)$ & $n^3 D (\log n + b)$ &$\mathbf{\log n + b}$\\
\bottomrule
\end{tabular}
\end{subtable}
\caption{\small Synoptic view of bounds for distributed consensus. Bounds on time complexity hold for all consensus functions, bounds on message complexity hold for locally-sensitive or extractive consensus functions, and, finally, bounds on byte  and storage complexity hold for locally-sensitive or extractive consensus functions that are hierarchically computable. S and D denote, respectively, sparse and dense graphs. The number of agents is denoted as $n$ and the number of communication links is $|E|$. $m$ is a tuning parameter assuming values in $\{1,\ldots,n\}$. The parameter $|E_c|$ is $|E_c|=O(\min(E,m^2))$. Time-varying networks are $D$-connected \cite{AO:10}.  For average-based algorithms, $\varepsilon$ is the convergence threshold for termination. We denote in bold face optimality results.\vspace{-.7cm}} 
\label{tab:consbounds}
}
\end{table*}

\section{Problem Setup}
\label{sec:background}
In this section we discuss the network model and we define the distributed consensus problem we will study in this paper. A simplified version of our model has first been introduced by the authors in \cite{FR-MP:13}.

\subsection{Agent model}

An agent in a robotic network is modeled as an input/output (I/O) automaton, i.e., a labeled state transition system able to send messages, react to received messages and perform arbitrary internal transitions based on the current state and on any messages received. 
A precise definition of I/O automaton is provided in \cite[pp. 200-204]{NL:96} and is omitted here in the interest of brevity. All nodes are identical except for a unique identifier (UID - for example, an integer). The time evolution of each node in the graph $G$ is characterized by two key assumptions:

\begin{itemize}
\item {\bf Fairness assumption}: the order in which transitions happen and messages are delivered is not fixed a priori. However, any enabled transition will \emph{eventually} happen and any sent message will \emph{eventually} be delivered.
\item {\bf Non-blocking assumption}: every transition is activated within $l$ time units of being enabled and every message is delivered within $d$ time units of being dispatched.
\end{itemize}
Essentially, the fairness assumption states that every node will have an opportunity to perform transitions, while the non-blocking assumption gives timing guarantees (but no synchronization). We refer the interested reader to \cite[pages 212-215]{NL:96} for a detailed discussion of these assumptions. We argue here that these are \emph{minimal} assumptions for most reliable real-world robotic networks. 

\subsection{Network model}\label{subsec:model}
A \emph{robotic network} with $n$ agents is modeled as a \emph{connected}, \emph{undirected} graph $G = (V,E)$, where $V = \{1,\ldots, n\}$ is the node set, and $E\subset V\times V$, the edge set,  is a set of \emph{unordered} node pairs modeling the availability of a communication channel. Two nodes $i$ and $j$ are neighbors if $(i, j)\in E$. The neighborhood set of node $i\in V$, denoted by $N_i$, is the set of nodes $j\in V$ neighbors of node $i$. Our model is \emph{asynchronous}, i.e., computation steps within each node and communication are, in general, asynchronous.

This paper focuses on \emph{static networks}, i.e., robotic networks where the edge set does not change during the execution of the algorithm. Given a node set with $n$ nodes, we will consider two classes of graphs:
\begin{itemize}
\item {\bf Sparse graphs}: that is graphs where the number of edges $|E|$ is less than $n\log n$.
\item {\bf Dense graphs}: that is graphs where the number of edges $|E|$ is larger than or equal to $n\log n$.
\end{itemize}

Henceforth, we will denote the set of sparse graphs as $\mathcal G_s$ and the set of dense graphs as $\mathcal G_d$. 
Note that for any connected graph with $n$ nodes, $n-1 \leq |E|\leq {n\choose 2}$.

From a practical standpoint, sparse and dense graphs manifest themselves in different robotic problems and give rise to different issues. In dense graphs (present e.g. in formation control and rendezvous problems) time complexity is typically not an issue; on the other hand, message and byte complexity have to be carefully kept under control to avoid excessive bandwidth utilization and minimize message collisions. Especially for large graphs, it is crucial to ensure that agents only communicate with a small subset of their neighbors, even if many are available. On the other hand, in sparse graphs (which typically manifest themselves in patrolling and deployment applications, where large inter-agent distances are desirable) message complexity is not an issue; efficient routing of information, on the other hand, is crucial to ensure good time performance.

\subsection{Model of computation}

At a general level, in this paper we focus on decision-making problems where each node $i$, $i\in\{1, \ldots, n\}$, in the robotic network is endowed with an initial value $x_i$ and should output the value of a function of all initial values. In other words, each agent, after exchanging messages (with any content) with its neighbors and performing internal state transitions, should output $f(x_1, \ldots, x_n)$ for some computable function $f$, referred to as \emph{consensus} function. In the reminder of this section we formalize the notions of consensus functions and of decentralized algorithms.

\subsubsection{Consensus functions}
In this paper we consider functions defined over \emph{totally ordered sets}, that is we assume that the initial conditions $x_i$ belong to a set $\mathcal X$ equipped with a binary relation (denoted with $\leq$) which is transitive, antisymmetric, and total.  Two sets of initial conditions $A=\{a_1, \ldots, a_n\}$ and $B=\{b_1, \ldots, b_n\}$ are said to be \emph{order-equivalent} if $a_i< a_j \Leftrightarrow b_i<b_j$.
The set of initial conditions $\mathcal C$ can be, for example, $\naturals$, $\reals$, and $\reals^d$ (in the last case the total order could be the lexicographic order).

A consensus function is a \emph{computable} 
function $f : \mathcal X^n \mapsto \reals$ that depends on \emph{all} its arguments. More precisely, for each element $x = (x_1, \ldots, x_n) \in \mathcal X^n$ and for all $i\in\{1, \ldots, n\}$ one can find elements $x_{i}^{(1)}\in \mathcal X$ and $x_{i}^{(2)}\in \mathcal X$ such that
\[
f(x_1, \ldots, x_{i}^{(1)}, \ldots, x_n)\neq f(x_1, \ldots, x_{i}^{(2)}, \ldots, x_n).
\]
Loosely speaking, such choice of consensus function implies that each node is needed for the collective decision-making process.

For some of the results presented in this paper, we will need the following  \emph{refinements} of the notion of consensus function:
\begin{itemize}
\item {\bf Locally-sensitive consensus function}: a consensus function that is sensitive to perturbations that preserve order. More precisely, let $\mathcal I = \{1,2, \ldots, n\}$ be the set of node indices and let $\sigma:\mathcal I \mapsto \mathcal I$ be a permutation 
(i.e., a bijective correspondence) over $\mathcal I$. Then, for each permutation $\sigma$ over $\mathcal I$ there exists an element $x\in \mathcal X^n$ that is ordered with respect to $\sigma$, i.e., $x_{\sigma(i)} \leq x_{\sigma(j)}$ for all $i<j$ and such that for all $i\in \{1, \ldots, n\}$ there exist $x_{\sigma(i)}^{(1)}, x_{\sigma(i)}^{(2)} \in \mathcal X$, $x_{\sigma(i)}^{(1)} \neq  x_{\sigma(i)}^{(2)}$ with the properties:
\begin{enumerate}
\item $x_{\sigma(i)}^{(1)}\leq x_{\sigma(j)}$ and  $x_{\sigma(i)}^{(2)}\leq x_{\sigma(j)}$ for all $i<j$,
\item $x_{\sigma(j)}\leq x_{\sigma(i)}^{(1)}$ and $x_{\sigma(j)}\leq x_{\sigma(i)}^{(2)}$ for all $j < i$,

\item $f(x_1, \ldots, x_{\sigma(i)}^{(1)}, \ldots x_n) \!\neq \!f(x_1, \ldots, x_{\sigma(i)}^{(2)}, \ldots x_n)$.
\end{enumerate}
(The first two properties ensure that $x_{\sigma(i)}^{(1)}$ and $x_{\sigma(i)}^{(2)}$ preserve the order of $x$, while the last property reflects local sensitivity.)
\item {\bf Extractive consensus function}: a consensus function such that for all $x=(x_1, \ldots, x_n)\in \mathcal X^n$ one has $f(x_1,\ldots, x_n)=x_j$ for some $j\in\{1, \ldots, n\}$.
\end{itemize}
The classes of locally-sensitive and extractive consensus functions are neither mutually exclusive nor collectively exhaustive (however, they represent,  arguably, a very broad class of consensus functions of interest to applications). Loosely speaking, locally-sensitive consensus functions model problems where the decision-making process depends \emph{continuously} on the initial values $\{x_i\}$, for example for average consensus 
$
f(x) = c^T\, x,
$ 
where $\mathcal X  =\reals^n$, $x\in \mathcal X$, and $c$ is a vector in $\reals^n$, or for 
distributed optimization
\[
f(x) = \argmin_{z\in \reals^n} \, \sum_{i=1}^n \, \varphi_i(z,x_i)
\]
where the objective function is \emph{parametrized} by $x_i$, under certain conditions on $\varphi_i$ (consider for instance $\varphi_i$ a positive-semidefinite quadratic form parametrized by $x_i$).
On the other hand, extractive consensus functions model leader-election problems or problems where it is desired to extract some statistics from the data, e.g., MAX and MIN. 

Finally, we introduce a \emph{representation} property for consensus functions that will be instrumental to deriving fundamental limitations of performance in terms of amount of information exchanged. A locally-sensitive or extractive consensus function is hierarchically computable if it can be written as  the composition of  a commutative and associative binary operator $\ast$, that is
\[
f(x_1, x_2,\ldots, x_n) = x_1 \ast x_2\ast \ldots\ast x_n.
\]
(The name is inspired by the observation that hierarchically computable functions can be computed with messages of small size on a \emph{hierarchical} structure such as a  tree). All examples of consensus functions mentioned above are indeed hierarchically computable.

\subsection{Model of communication}
Nodes can communicate with their neighbors according to two communication schemes: \emph{directional} and \emph{local broadcast}. In the directional communication scheme a node sends messages to each neighbor individually. This is the case when nodes in the network are equipped with narrow-band, high-gain mechanically or electronically steerable antennas. In the local broadcast communication scheme, a node sends a message to all its neighbors simultaneously. This is the typical case for nodes equipped with omnidirectional antennas.

\subsection{Distributed algorithms}
A distributed algorithm for a robotic network is, simply, a collection of \emph{local} algorithms, one for each node of the network (of course, nodes can exchange messages). Nodes execute the same logical code. Each node is initialized with an initial condition $x_i \in \mathcal X$.
A distributed  algorithm correctly computes a given function if, given an input $x\in \mathcal X^n$, \emph{each} node outputs the correct value of the function $f(x_1, \ldots, x_n)$, and terminates \cite{NL:96} in \emph{finite} time. Specifically, the algorithm should output the correct value in at most  $R$ rounds (for synchronous executions) or after $R\, (l+d)$ time units (for asynchronous executions). $R$ can be an arbitrarily large number, hence this assumption is not limiting from a practical standpoint; however, it will be key to deriving some of our results.
Note that, for asynchronous executions, termination is not simultaneous. 

A particular class of distributed algorithms that will be instrumental to derive some of the results is represented by comparison-based algorithms, defined next. 
\begin{definition}[Comparison-based algorithms (\cite{GF-NL:87, MS:85})]
\label{def:compbased} 
A distributed algorithm is comparison-based if each local algorithm manipulates the subset of (totally ordered) initial values that are locally known only via the three Boolean operators $<$,$ >$, and $=$ (recall that the set of initial values is a totally ordered set).
Accordingly, all internal transitions (including message generation) only depend on the \emph{order} of the initial values known to the local algorithm, as opposed to their numerical value. 
\end{definition}

\subsection{Complexity measures}\label{subsec:com}
The following definitions naturally capture the notions of time and communication complexity and are widely used in the theory of distributed algorithms  \cite{NL:96}.

Let $\mathcal G$ be a set of graphs with node set $V=\{1,\ldots, n\}$ (we are specifically interested in the class of sparse graphs $\mathcal G_s$ and in the class of dense graphs $\mathcal G_d$). For a given graph $G\in \mathcal G$, let $\mathcal F(a, x, G)$ be the set of \emph{fair executions}  for an algorithm  $a\in \mathcal A$ and a set of initial conditions $x \in \mathcal X^{n}$  (a fair execution is an execution of an algorithm that satisfies the fairness and non-blocking assumptions stated above).

\subsubsection{Time complexity}  To measure execution time, we assume that a distributed algorithm starts at time $t=0$. Time complexity is defined as the infimum worst-case (over initial values and fair executions) completion time of an algorithm. Rigorously, the time complexity for a given consensus function $f$ with respect to the class of graphs $\mathcal G$ is
\[
\textrm{TC}(f, G):=\inf_{a\in \mathcal A} \, \sup_{G\in \mathcal G}\, \sup_{x \in \mathcal X^{|G|}} \, \sup_{\alpha \in \mathcal F(a, x,G)} \, T(a, x, \alpha, G),
\]
where $T(a, x, \alpha, G)$ is the first time when all nodes have computed the correct value for the consensus function $f$ and have stopped.
The order of the inf-sup operands in the above expression is naturally induced by our definitions. By dropping the leading $\inf_{a\in \mathcal A}$, one recovers the time complexity of a given algorithm $a$ for a given consensus function $f$. In our asynchronous setting, time complexity is expressed in multiples of $l+d$, defined in section \ref{subsec:model} (see also  \cite{NL:96}). We will henceforth refer to $(l+d)$ as a \emph{time unit}. Note that $(l+d)$ is a (tight) upper bound on the actual time required to generate and deliver a message.
\subsubsection{Message complexity for directional communication} Message complexity is similarly defined as the infimum worst-case (over initial values and fair executions) \emph{number} of messages exchanged by an algorithm before completion. Rigorously, the message complexity for a given consensus function $f$ with respect to the class of graphs $\mathcal G$ is
\[
\text{MC}(f,\mathcal G)=\inf_{a\in \mathcal A} \, \sup_{G\in \mathcal G}\,\sup_{x \in \mathcal X^{|G|}} \, \sup_{\alpha \in \mathcal F(a, x,G)} \, M(a, x, \alpha,G),
\]
where $M(a, x, \alpha, G)$ is the number of messages exchanged between time $t=0$ and  time $t = T(a, x, \alpha,G)$.

It is important to note that the \emph{type} of messages exchanged depends on the algorithm. In average-based consensus algorithms, nodes typically exchange their state, a real number. In algorithms such as the well-known Gallager, Humblet and Spira (GHS) algorithm \cite{RGG-PAH-PMS:83}, nodes exchange a wide range of logical commands establishing hierarchical relationships, informing neighbors about the progress of the algorithm, and requiring them to perform edge searches  \cite{ BAw:87}. In flooding algorithms \cite{NL:96}, a single message may contain information from up to $n-1$ nodes. However, as far as message complexity is concerned, each message \emph{counts the same}, regardless of its type and size.
\subsubsection{Byte complexity for directional communication} 
In many instances, message size plays a critical role in the energy needed for information transmission. To capture this aspect, in this paper we define byte complexity as the infimum worst-case (over initial values and fair executions) \emph{overall size} (in bytes) of all messages exchanged by an algorithm before its completion. Rigorously, the byte complexity for a given consensus function $f$ with respect to the class of graphs $\mathcal G$ is
\[
\textrm{BC}(f,\mathcal G):=\inf_{a\in \mathcal A} \, \sup_{G\in \mathcal G}\,\sup_{x \in \mathcal X^{|G|}} \, \sup_{\alpha \in \mathcal F(a, x,G)} \, B(a, x, \alpha,G),
\]
where $B(a, x, \alpha, G)$ is the overall size (in bytes) of all messages exchanged  between time $t=0$ and  time $t = T(a, x, \alpha,G)$.

\subsubsection{Message and byte complexity for local broadcast communication}
The definitions of message and byte complexity for local broadcast communication parallels the definitions of message and byte complexity for directional communication.
Rigorously, the broadcast message complexity for a given consensus function $f$ with respect to a class of graphs $\mathcal{G}$ is
\[
\textrm{bMC}(f,\mathcal{G})=\inf_{a\in\mathcal{A}}\sup_{G\in\mathcal{G}}\sup_{x\in\mathcal X^{|G|}}\sup_{\alpha\in\mathcal{F}(\alpha,x,G)} bM(a,x,\alpha,G)
\]
where $bM(a,x,\alpha,G)$ is the overall number of \emph{broadcast} messages exchanged between time $t=0$ and  time $t = T(a, x, \alpha)$ (a broadcast message is a message sent by a node to \emph{all} its neighbors). Analogously, the broadcast byte complexity for a given consensus function $f$ with respect to a class of graphs $\mathcal{G}$ is
\[
\textrm{bBC}(f,\mathcal{G})=\inf_{a\in\mathcal{A}}\sup_{G\in\mathcal{G}}\sup_{k\in\mathcal X^{|G|}}\sup_{\alpha\in\mathcal{F}(\alpha,x,G)} bB(a,x,\alpha,G)
\]
where $bB(a,x,\alpha,G)$ is the size (in bytes) of all \emph{broadcast} messages exchanged between $t=0$ and $t = T(a, x, \alpha)$.

\subsection{Discussion}

It is of interest to compare our model of distributed consensus with the models for the consensus problem developed, respectively, by the Computer Science community and the Controls and Robotics community. In the Computer Science community, distributed consensus is typically defined as the task of computing via a distributed algorithm \emph{any} function of a set of initial conditions such that the following three properties are fulfilled: \emph{agreement} (no two processes decide on different values), \emph{validity} (in absence of failures, if all agents start with the same value, then every agent decides on that value) and \emph{termination} (all processes eventually decide) \cite{NL:96}. In the Controls and Robotics community, on the other hand, consensus is essentially a synonym for \emph{average-based} consensus, where agents compute an asymptotic approximation of a weighted average of their initial conditions via local communication \cite{AO:10,AJ-JL-ASM:02, ROS-RMM:03c,ROS:07,WR-RWB-EMA:07,MdW-BC:09}. Our model, thus, is more restrictive than the model considered in the Computer Science community (since the consensus function is required to fulfill some mild requirements), while it is (significantly) more general than the model considered in the Controls and Robotics community (since average-based consensus is a particular example of locally-sensitive consensus function defined over $\reals^n$).

We mention that we also studied \emph{storage complexity}, defined as the infimum worst-case (over initial values and fair executions) storage size required by every agent executing the consensus algorithm. We do not discuss this complexity notion here due to space limitation and because in most cases it is not a bounding factor. However, we report our results in the synoptic table (Table \ref{tab:consbounds}).

In the remainder of the paper we discuss fundamental limitations of performance of the distributed consensus problem, in terms of fundamental scalings of the different complexity measures with respect to the network size (i.e., the number of nodes $n$). We also discuss algorithms that, in many cases, recover such asymptotic bounds. Our asymptotic notation (e.g., $O(g(n))$ or $\Omega(g(n))$ is standard.

\section{Lower Bounds on Achievable Performance for Distributed Consensus}
\label{sec:lbs}
In this section we present lower bounds for the complexity measures introduced in Section \ref{subsec:com}. We will discuss the tightness of these bounds in Section \ref{sec:upper}.

\subsection{Time complexity}
A lower bound on time complexity can be obtained rather easily. 
\begin{proposition}[Lower bound on time complexity]\label{prop:lbtime}
 For a given consensus function $f$ and class of graphs $\mathcal G$ with $n$ nodes,  $\textrm{TC}(f, \mathcal G) \in \Omega(n)$. 
\end{proposition}
\begin{proof} 
By contradiction. 
Let us assume that there exists a consensus algorithm $a$ that terminates in $o(n)$ time units for all graphs $G\in\mathcal{G}$, initial conditions $x\in\mathcal{X}$ and executions $\alpha$.
We restrict our analysis to synchronous executions of the algorithm (since synchronous executions are a special case of asynchronous executions, a lower bound with respect to the former is also a bound for the asynchronous case). We also consider a specific graph $G$ where the maximal distance between any two pairs of nodes is $\text{Diam}(G)=\Theta(n)$ (the \emph{line} graph is an example of one such graph). 
Then there exist two nodes $u$, $v$  such that $n$ time units are required for \emph{any} information from agent $u$ to reach agent $v$ and vice versa.
Now, consider two executions of the consensus algorithm that only differ in the initial value of agent $v$. Rigorously, we consider two sets of initial values $x^{(1)}, x^{(2)}$ with $x^{(1)}_v\neq x^{(2)}_v$ and $x^{(1)}_k=x^{(2)}_k\,\forall k\neq v$.
Since the algorithm terminates in fewer than $n$ time units, agent $u$ does not hear any information from agent $v$ in either execution: therefore its state is identical at the end of both executions. Then agent $u$ decides on the same consensus value in both executions. However, $f(x^{(1)})\neq f(x^{(2)})$ since $x^{(1)}_v\neq x^{(2)}_v$: we have reached a contradiction.
\end{proof}

\subsection{Message complexity}
\label{sec:mc}
In this section we restrict our attention to either locally-sensitive or extractive consensus functions. Our strategy is to find first a lower bound for ``dense" graphs and then a lower bound for "sparse" graphs.
We start with the former case.

\begin{proposition}[Lower bound on message complexity for dense graphs]
\label{lemma:ccfull}
For a given locally-sensitive or extractive consensus function $f$ ,  $\textrm{MC}(P, \mathcal G_d) \in \Omega(|E|)$.
\end{proposition}
\begin{proof}
Computation of any consensus function $f$ requires that \emph{at least} one message is sent along \emph{every} edge of a spanning subgraph of the network.  If this were not true, there would exist two subsets of the nodes $V_1\subset V$ and $V_2=V\setminus V_1$ s.t. no messages are exchanged between nodes in $V_1$ and in $V_2$. Then, nodes in $V_1$ would have no information about the initial values of nodes in $V_2$ and vice versa. Since $f$ depends on all initial values, this leads to a contradiction. Now, it can be shown that \emph{any} computation problem that requires use of a \emph{spanning subgraph} (i.e., at least one message sent along each of its edges) may require $|E|-1$ messages (and therefore $\Omega(|E|)$ messages) on a certain class of \emph{almost complete} graphs \cite{EK-SM-SZ:84}. This concludes the proof.
\end{proof}

We now turn our attention to a lower bound that becomes tight for sparse graphs. We first consider comparison-based algorithms; we will relax this assumption in Proposition \ref{lemma:ccring_noncomp}.

\begin{lemma}[Lower bound on message complexity for sparse graphs and comparison-based algorithms]
\label{lemma:ccring_comp}
Let $f$ be a locally-sensitive or extractive consensus function.
Let $\mathcal A_c$ be the set of comparison-based algorithms that solve the distributed consensus problem and assume that one minimizes message complexity over the set $\mathcal A_c$ (we denote the result of such minimization as $\textrm{MC}(f, \mathcal G_s)_{|\mathcal A = \mathcal A_c}$). Then  $\textrm{MC}(f, \mathcal G_s)_{|\mathcal A = \mathcal A_c} \in \Omega(n\log n)$.
\end{lemma}
\begin{proof}
Consider the restriction  to synchronous executions  (since synchronous executions are a special case of asynchronous executions, a lower bound with respect to synchronous executions translates into a bound for the more general  asynchronous case). 

The proof  is inspired by  \cite{GF-NL:87} and relies on the notion of $c$-symmetric rings. Consider a graph with a ring topology (i.e., the $n$ nodes are lined up along a circle). A segment $S$ on the ring is a sequence of consecutive nodes in the ring, in clockwise order. Two segments of the same length are said order-equivalent if the ordered vector of initial conditions of their respective nodes are order-equivalent. Let $c$ be a positive constant. A ring is $c$-symmetric if, for every $l \in \naturals$ such that $\sqrt{n}\leq l \leq n$, and for every segment $S$ of length $l$, there are at least $\lfloor cn/l \rfloor$ segments in the ring that are order-equivalent to $S$ (including $S$ itself).  An example is shown in Fig. \ref{fig:csymmetric}.

We now study the message complexity of comparison-based algorithms on $c$-symmetric rings. To this purpose, without loss of generality\footnote{Any  consensus algorithm on a ring can be simulated by an algorithm in this class, since we assume no bound on the nodes' computational power and no limitations on their internal transitions; therefore any lower bound on this class of algorithms applies to \emph{all} consensus algorithms on rings. },  we consider comparison-based algorithms where at each synchronous time step (recall that we are considering \emph{synchronous} executions)
a node decides whether to send a message to its right  neighbor, whether to send a message to its left  neighbor, and whether to stop execution and decide on a consensus value. Every received message is stored in the receiver node's state. Every sent message contains the sender's entire state. Nodes can perform arbitrary internal transitions and have unlimited computational power. At each time step, the state of a node contains its initial value, its UID, and the history of messages exchanged with the neighbors. The initial conditions \emph{include the nodes' UIDs}, endowed with a total ordering.

It can be shown that (see \cite{GF-NL:87}) (i) for any $n$, there exists a set of initial conditions such that there exists a $c$-symmetric ring  for some $ c>0$, (ii) if a comparison-based algorithm exchanges $o (n\log n)$ messages on a $c$-symmetric ring, then every node receives information from nodes within distance $k$ where $k<n/2$ (hence, from a \emph{subset} of all nodes), and (iii) if a comparison-based algorithm exchanges $o(n\log n)$ messages on a $c$-symmetric ring, then every node $i$'s state is order-equivalent to another node $j$'s state at the end of the execution. More precisely, the $k$-neighborhoods of agents $i$ and $j$ (defined as the set containing the node and the $2k$ neighbors closest to it) contain agents with UIDs in identical order: such neighborhoods (shown in the example in Fig. \ref{fig:csymmetric}) can not be distinguished by a comparison-based algorithm.
\begin{figure}[h]
\centering
\includegraphics[width=.34\textwidth]{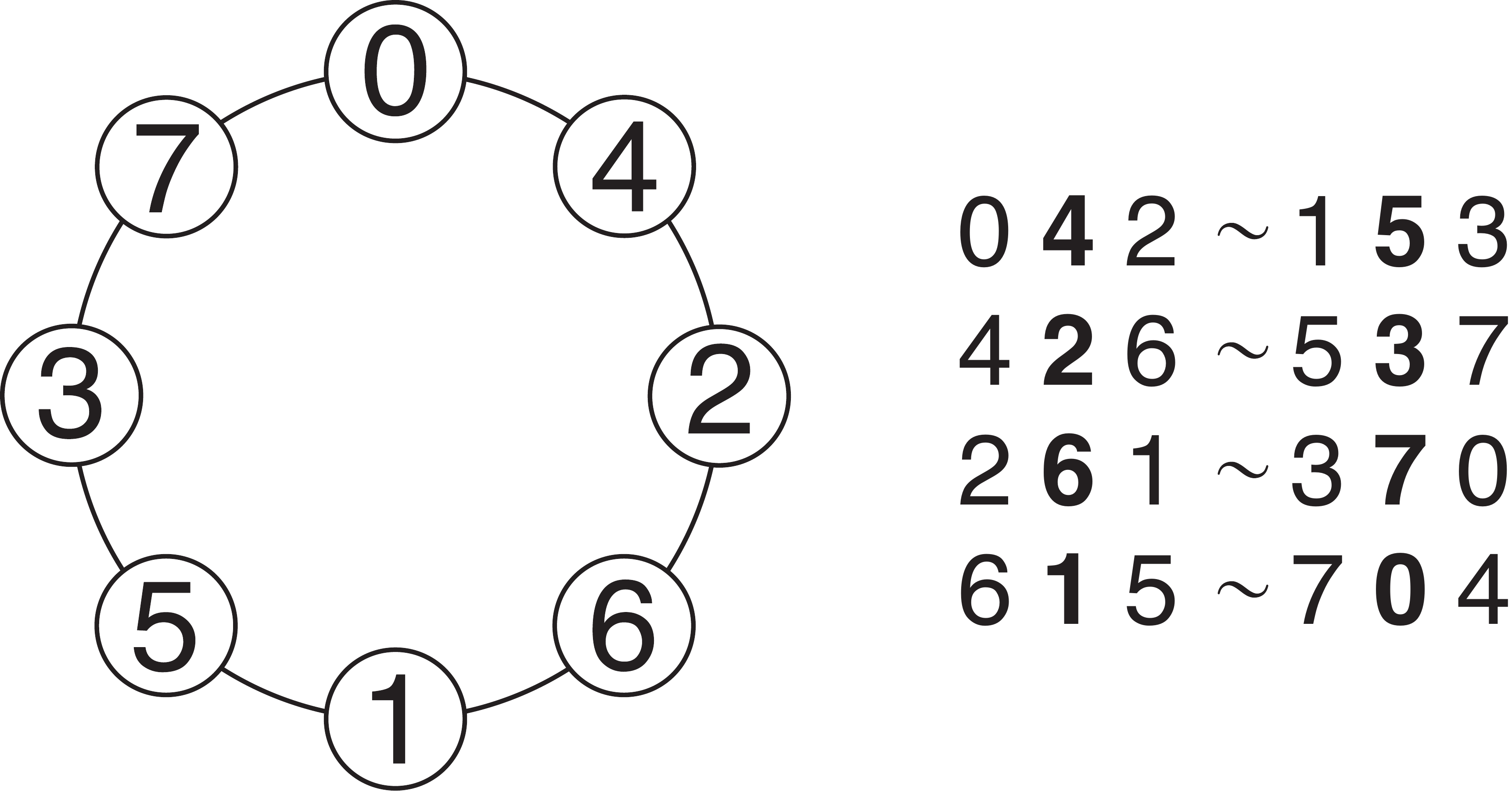}
\caption{\small A 1/2-symmetric ring. Each segment of length $l\leq4$ (and, in particular, $l=3$) is order-equivalent to another segment. Thus, each agent's 2-neighborhood is order-equivalent to another agent's.\vspace{-0.35cm}}
\label{fig:csymmetric}
\end{figure}
This implies that the leader election problem has a message complexity $\Omega(n\log n)$  \cite{GF-NL:87}, since after $o (n\log n)$ message at least two agents are in states that are indistinguishable by a comparison-based  algorithm.

We now apply these results to  distributed consensus problems with locally-sensitive or extractive consensus functions. The lower bound $\Omega(n\log n)$ on message complexity   directly applies to extractive consensus functions: any distributed consensus algorithm capable of extracting the initial value of a node with $o(n\log n)$ messages would also solve the leader election problem with $o(n\log n)$ messages, which would contradict the aforementioned result for leader election. 

Consider, now, locally-sensitive consensus functions. We proceed by contradiction, assuming that there exists a comparison-based algorithm that  solves the  consensus problem with a  locally-sensitive consensus function with message complexity  $o(n\log n)$.  Consider a given set of initial conditions, $x$, and let the nodes compute the value of the consensus function. Then, consider a set of initial conditions $x^{\prime}$ identical to $x$ except that one of the node's initial values (and therefore the overall consensus value) is perturbed  without changing the overall ordering of UIDs and node values (this is possible by the assumption of locally-sensitive consensus function). We next show that  after $o(n\log n)$ messages at least one node would output the same consensus value it computed for initial condition $x$ -- a contradiction.

Let the nodes be arranged in a ring and let $x^{(1)}$ be an initial condition such that the ring is $c$-symmetric. By the contradiction hypothesis, after $o(n\log n)$ messages are exchanged, every node's state is order-equivalent to at least one other node's state and each node correctly outputs $f(x^{(1)})$ - call this execution $\alpha^{(1)}$. Note that, by fact (ii) above,  each node has only received information from $2k+1<n$ neighbors, including itself. Consider a pair $u$, $v$ of nodes in order-equivalent states: there exists one node $w$ that belongs to the $k$-neighborhood of $u$ but does not belong to the $k$-neighborhood of $v$. Consider now an initial condition $x^{(2)}$ identical to $x^{(1)}$ except that  the initial value of $w$ is perturbed without changing the overall order of agents' values and UIDs, and such that $f(x^{(1)})\neq f(x^{(2)})$ (this is always possible when the consensus function is locally-sensitive). Given initial condition $x^{(2)}$, under any execution $\alpha^{(2)}$ with $o(n\log n)$ messages, $v$'s state will be identical (and not just order-equivalent) to the state under execution $\alpha^{(1)}$, and therefore $v$ will output  $f(x^{(1)})$ as its consensus value (since by the contradiction hypothesis the algorithm terminates after $o(n\log n)$ messages). This is a contradiction.
 \end{proof}

The bound in Lemma \ref{lemma:ccring_comp} is instrumental to derive the desired lower bound over the much more general class of distributed consensus algorithms. Such bound requires that the set of initial values, i.e., $\mathcal X$ is ``large enough". This requirement is automatically satisfied whenever $\mathcal X$ has infinite cardinality.

\begin{proposition}[Lower bound on message complexity for sparse graphs]
\label{lemma:ccring_noncomp}
Let $f$ be a locally-sensitive or extractive consensus function. Then there exists a function $\psi(n,R)$ such that, if the cardinality of $\mathcal X$ is greater than or equal to $\psi(n,R)$, then $\textrm{MC}(f, \mathcal G_s) \in \Omega(n\log n)$.
\end{proposition}
\begin{proof}
The key idea (conceptually identical to that in \cite{GF-NL:87} and \cite{MS:85}) is to  show that, perhaps surprisingly,  if the cardinality of $\mathcal X$  is larger than a (very large) finite number, any distributed consensus algorithm in $\mathcal A$  executes on a small subset of $\mathcal X$ in a way that it is indistinguishable from the execution of a comparison-based algorithm. Lemma \ref{lemma:ccring_comp} then applies and the claim follows.

As in the proof of Lemma \ref{lemma:ccring_comp}, without loss of generality, we consider synchronous executions and distributed consensus algorithms  (referred to as \emph{elementary} algorithm) where at each time step a node decides whether to send a message to its right  neighbor, whether to send a message to its left  neighbor, and whether to stop execution and decide on a consensus value. Every received message is stored in the receiver node's state. Every sent message contains the sender's entire state. Nodes can perform arbitrary internal transitions and have unlimited computational power. At each time step, the state of a node contains its initial value, its UID, and the history of messages exchanged with the neighbors. The initial conditions \emph{include the nodes' UIDs} (totally ordered according to some binary relation $\leq$). 

We introduce the definition of \emph{indistinguishable} initial values.
Consider two sets of initial conditions $x^{(1)}$ and $x^{(2)}$, whose elements are arranged in increasing order, that is $x^{(1)}_i \leq x^{(1)}_j $ and $x^{(2)}_i \leq x^{(2)}_j $ for $i\leq j$. Let $\sigma$ be a permutation over the set of indexes $\mathcal I=\{1,\ldots, n\}$. We say that  $x^{(1)}$ and $x^{(2)}$ are \emph{indistinguishable} with respect to an algorithm $a\in \mathcal{A}$ if, for \emph{any} permutation $\sigma$, the trace\footnote{Informally, the trace is the history of an execution the algorithm: for each time step, it records all agents' states  and all messages exchanged. For a formal definition, we refer the reader to \cite{NL:96}.} of the execution with initial values $x^{(2)}$ (with indices permuted according to $\sigma$) can be obtained from the trace of the algorithm with initial values $x^{(1)}$ (with indices permuted according to $\sigma$) merely by substituting every occurrence of an element of $x^{(2)}$ with the element of corresponding order in $x^{(1)}$. 

We claim that, if the set $\mathcal X$ of possible initial values is large enough, there exists a set $\mathcal W \subseteq \mathcal X$ with $|\mathcal W|\geq 2n-1$ such that \emph{any} two $n$-subsets $\mathcal U \subset \mathcal W$ of size $|\mathcal U|=n$ are indistinguishable with respect to a given consensus algorithm $a\in \mathcal A$. This claim follows from Ramsey's theorem \cite[Theorem B]{FPR:30}. 
Specifically, color every subset  $\mathcal U\subset \mathcal X$ of size $|\mathcal U|=n$ so that indistinguishable sets  share the same color. For any elementary algorithm, there are finitely many set of indistinguishable initial conditions. This is due to the facts that (i) there are finitely many permutations $\sigma$ (specifically, $n!$), (ii) at each time step each node can make a finite number of decisions (specifically, 8), and (iii) the algorithm must terminate within a number of rounds equal to  $R$. Then, by Ramsey's theorem, there exists a number $\phi(n,R)$ such that, if $|\mathcal X|\geq \phi(n,R)$, then there exists at least one set $\mathcal W \subset \mathcal X$ of size $|\mathcal W|=2n-1$ such that all its $n$-subsets $\mathcal U\subset \mathcal W$, $|\mathcal U|=n$, share the same color. Note that $\phi(n,R)$ does not depend on the specific algorithm under consideration.

We next claim that there is a set $\bar {\mathcal U} \subset \mathcal W$ of size $|\bar {\mathcal U}|=n$ such that any elementary algorithm  behaves like a comparison-based algorithm (i.e., one can find a comparison-based algorithm yielding identical executions) on initial conditions taken from $\bar {\mathcal U}$. Specifically, take $\bar {\mathcal U}$ as the set of the $n$ lowest values in $\mathcal W$. We claim that \emph{any} two order-equivalent sets of size $m\leq n$ in $\bar {\mathcal U}$ (corresponding to $m$-neighborhoods of agents executing the distributed algorithm) yield identical executions, thus implying that the algorithm effectively emulates a comparison-based algorithms on initial values from $\bar {\mathcal U}$. To prove the claim, consider two sets $\bar {\mathcal U}_1$, $\bar {\mathcal U}_2\in \bar {\mathcal U}$ of size $m$. Now append to $\bar {\mathcal U}_1$ and $\bar {\mathcal U}_2$ the same $n-m$ elements of $\mathcal W\setminus \mathcal U$. The two  resulting $n$-sets belong to $\mathcal W$ and elements of $\bar {\mathcal U}_1$ and $\bar {\mathcal U}_2$ appear in the same position in both $n$-sets: therefore they are indistinguishable and, in particular, whenever the states of two nodes $u$ and $v$ contain sets $\bar {\mathcal U}_1$ and $\bar {\mathcal U}_2$, respectively, such nodes will output the same value for the consensus function.

The proof is then completed by using executions of a (non-comparison-based) elementary algorithm on $\bar {\mathcal U}$ to ``construct" a comparison-based algorithm. Specifically, we construct a   comparison-based algorithm whose transitions are identical to the transitions of a given elementary algorithm on $\bar{\mathcal U}$, based on the order of the elements in $\bar{\mathcal U}$. Since no comparison-based algorithm can solve the consensus problem with  $o(n\log n)$ messages, the claim follows.
\end{proof}

We stress the fact that the assumptions of Proposition \ref{lemma:ccring_noncomp} are satisfied  whenever the set of initial conditions has infinite cardinality (e.g., $\mathcal X = \reals$ or $\mathcal X = \naturals$).

\subsection{Byte complexity}
\label{ssec:bc}
To prove bounds on byte complexity, we must be a little bit more specific about the content of the messages. In particular, we will assume that messages carry the UID of the sender and/or of the receiver. In practice, virtually all wireless communication protocols require each message to carry a UID identifying the sender. In addition, in non-broadcast communication protocols (as those considered in this section), messages also need to carry a receiver UID.  The proof of Propositions are omitted: they follow quite easily from Propositions \ref{lemma:ccfull} and \ref{lemma:ccring_noncomp} (and the fact that a UID requires $\log n$ bytes to be transmitted). 

 \begin{proposition}[Lower bound on byte complexity for sparse networks]
 \label{lemma:bcring}
Assume messages carry the sender and/or the receiver UIDs. Let $f$ be a hierarchically computable and either locally-sensitive or extractive consensus function. There exists a function $\psi(n,R)$ such that, if the cardinality of $\mathcal X$ is greater than or equal to $\psi(n,R)$, then  $\textrm{BC}(f, \mathcal G_s) \in \Omega\Bigl ((n\log n)\log n+n\, b \Bigr)$, where $b$ is the size (in bytes) of an initial condition in $\mathcal X$.
 \end{proposition}

\begin{proposition}[Lower bound on byte complexity for dense networks]
\label{lemma:bcfull}
Assume messages carry the sender and/or the receiver UIDs.
Then  $\textrm{BC}(f, \mathcal G_d) \in \Omega\Bigl (|E|\log n+n\, b \Bigr)$, where $b$ is the size (in bytes) of an initial condition in $\mathcal X$.
\end{proposition}

\subsection{Lower bounds for local broadcast algorithms}
The lower bounds on message complexity in Sections \ref{sec:mc} and  \ref{ssec:bc} are derived under the assumption of directional communication (the lower bounds on time  complexity is, instead, general). This section adapts those bounds to the case of local broadcast communication. The proofs for Propositions \ref{prop:byte1} and \ref{prop:byte2} are omitted: they are a simple consequence of the fact that on ring topologies local broadcasts only offer a twofold improvement in message complexity. Note that, in this case, we do \emph{not} make a distinction between dense and sparse graphs.

\begin{proposition}[Lower bound on broadcast message complexity]\label{prop:byte1}
Let $f$ be a locally-sensitive or extractive consensus function and $\mathcal G$  be a set of graphs with $n$ nodes.There exists a function $\psi(n,R)$ such that, if the cardinality of $\mathcal X$ is greater than or equal to $\psi(n,R)$, then $\textrm{bMC}(f, \mathcal G) \in \Omega(n\log n)$.
\end{proposition}

\begin{proposition}[Lower bound on broadcast byte complexity]\label{prop:byte2}
Assume messages carry the sender and/or the receiver UIDs. Let $f$ be a locally-sensitive or extractive consensus function and $\mathcal G$  be a set of graphs with $n$ nodes. There exists a function $\psi(n,R)$ such that, if the cardinality of $\mathcal X$ is greater than or equal to $\psi(n,R)$, then   $\textrm{bBC}(f, \mathcal G) \in \Omega(n\log^2 n+nb)$.
\end{proposition}
Note that the lower bound on byte complexity for local broadcast schemes is lower than the corresponding lower bound for directional communication, as one might intuitively expect.

\section{Tightness of Lower Bounds on \\Achievable Performance}\label{sec:upper}
In this section we study the tightness of the bounds derived in Section \ref{sec:lbs}. 

\subsubsection{Time complexity}
The bound on time complexity is tight and is achieved by a \emph{flooding algorithm}, which repeatedly transmits its initial value and \emph{all} received information to its neighbors (this result is well-known in the context of leader election -- its extension to our setting is straightforward).

\begin{proposition}[Tightness of time complexity]
\label{lemma:tctight}
For a given consensus function $f$ and class of graphs $\mathcal G$ with $n$ nodes,  $\textrm{TC}(f, \mathcal G) \in \Theta(n)$. 

\end{proposition}
\begin{proof}
 In a flooding algorithm, information travels from a node $v$ to any node at distance $k$ from $v$ in no more than $k\, (l+d)$ time units (recall that, within our model, by the fairness and non-blocking assumptions, each node executing a flooding algorithm will transmit a message at least once every $l+d$ time units). Since the distance between any pair of nodes in $G$ is smaller than or equal to $n$, and since each node can correctly compute the value of $f$ once it has knowledge of all initial values, one can conclude that a flooding algorithm has time complexity $O(n)$. Comparison with the the lower bound in Proposition \ref{prop:lbtime} immediately leads to the claim.
\end{proof}

\subsubsection{Message complexity}
Remarkably, a slight variant of the GHS algorithm \cite{RGG-PAH-PMS:83, BAw:87} (which builds a minimum spanning tree) achieves message optimality both for dense and for sparse graphs.
\begin{proposition}
\label{lemma:mctight}
Let $f$ be a locally-sensitive or extractive consensus function. Then $\textrm{MC}(f,\mathcal G_d)\in \Theta(|E|)$. Furthermore, there exists a function $\psi(n,R)$ such that, if the cardinality of $\mathcal X$ is greater than or equal to $\psi(n,R)$, then $\textrm{MC}(f, \mathcal G_s) \in \Theta(n\log n)$.
\end{proposition}
\begin{proof}
Consider the following variant of the GHS algorithm. First,  a rooted minimum spanning tree (MST) is constructed by executing the GHS algorithm. This operation requires $O(n\log n + |E|)$ messages. Note that at the end of the GHS algorithm the node that is the root of the MST is aware of this fact. The root node then requests from all nodes their initial values with a tree broadcast. After a node is contacted, it waits until all descendants (if any) have sent it their initial value; it then forward its initial value and its descendants' to its parent. Finally, the root computes the consensus function and sends the consensus value to all nodes. Given the tree structure, tree broadcasts and information collection require exactly $n-1$ messages each. The claim then follows.
\end{proof}

\subsubsection{Byte complexity}
To prove tightness of the byte complexity bound, we need to assume that the locally-sensitive or extractive consensus function is hierarchically computable. 
\begin{proposition}
\label{lemma:bctight}
Assume messages carry the sender and/or the receiver UIDs. Let $f$ be a locally-sensitive or extractive consensus function that is hierarchically computable. Then $\textrm{BC}(f, \mathcal G_d) \in \Theta(|E|\log n + nb)$. Furthermore, there exists a function $\psi(n,R)$ such that, if $|\mathcal X| \geq \psi(n,R)$, then   $\textrm{BC}(f, \mathcal G_s) \in \Theta((n\log n)\log n + nb)$.
\end{proposition}
\begin{proof}
Consider the same variant of the GHS algorithm introduced in the proof of Lemma \ref{lemma:mctight}. As discussed, the GHS algorithm computes a rooted MST in $O(n \log n + |E|)$ messages, and the consensus function is computed with further $O(n)$ messages. Each message exchanged by the GHS algorithm during the construction of the tree has size $O(\log n)$ \cite{RGG-PAH-PMS:83}. Furthermore, under the assumption that the consensus function is hierarchically computable, each message exchanged transmitted along the tree has size $O(b)$ and there are $O(n)$ such messages. The claim then follows.
\end{proof}

\subsubsection{Tightness of bounds for  local broadcast communication}
The study of the tightness of the bounds for local broadcast communication hinges upon a slightly more intricate variation of the GHS algorithm. Specifically, the message complexity of the GHS algorithm is $O(n\log n + |E|)$; the $|E|$ factor is due exclusively to challenge-reject message pairs exchanged by nodes during the search for a minimum weight outgoing edge (MWOE). In the MWOE search phase, each node contacts the neighbor connected to its lowest-weight edge: the neighbor's reply is positive or negative depending on the two nodes' group IDs and can be delayed based on the two nodes' levels (we refer the interested reader to \cite{RGG-PAH-PMS:83} for an in-depth definition of this terminology).
Consider, now, a broadcast protocol, and let each node simply \emph{broadcast} its level and group ID every time these are updated (i.e., at most once per level). We remark that a node can assume at most $\log n$ levels during execution. Neighbor nodes locally record the broadcasts they receive and look them up when looking for the MWOE. It is easy to see that such an algorithm \emph{emulates} the execution of the GHS algorithm (and, hence, inherits its correctness).

\begin{proposition}[Broadcast message complexity of consensus]
\label{lemma:broadcastmc}
Let $f$ be a locally-sensitive or extractive consensus function and $\mathcal G$ the set of graphs with node set $n$. There exists a function $\psi(n,R)$ such that, if the cardinality of $\mathcal X$ is greater than or equal to $\psi(n,R)$, then $\textrm{bMC}(f, \mathcal G) \in \Theta(n\log n)$.
\end{proposition}
\begin{proof}
Consider a distributed consensus algorithm that first \emph{emulates} the GHS algorithm to construct a rooted spanning tree (as discussed above) and then performs the sequence of initial value requests and routings  outlined in the proof of Proposition \ref{lemma:mctight}. The overall number of broadcasts exchanged per level during the emulation of the GHS algorithm is  $\Theta(n)$, since each agent only updates its level and group ID once per level, and the number of levels is $O(\log n)$. Given that the initial values requests and routings require $O(n)$ broadcast operations, the algorithm has a broadcast message complexity of $O(n\log n)$: the claim follows.
\end{proof}

The tightness of the bound on byte complexity follows immediately from Proposition \ref{lemma:broadcastmc} and its proof is omitted in the interest of brevity.
\begin{proposition}[Broadcast byte complexity of consensus]
\label{lemma:broadcastbc}
Assume messages carry the sender or the receiver UIDs. Let $f$ be a locally-sensitive or extractive consensus function and $\mathcal G$  a the set of graphs with $n$ nodes. There exists a function $\psi(n,R)$ such that, if the cardinality of $\mathcal X$ is greater than or equal to $\psi(n,R)$, then   
$\textrm{bBC}(f, \mathcal G) \in \Theta(n\log^2 n + nb))$.
\end{proposition}

\section{Discussion and conclusions}
\label{sec:conclusions}
In Table \ref{tab:consbounds} we provide a synoptic view of our results.  Table \ref{tab:consbounds} also includes results for $D$-connected networks, that is networks  where the edge set is time-varying and there exists a constant $D\in \reals_{>0}$ (possibly unknown to the nodes) such that the union of all edges appearing in the time interval $[t, t+D)$ constitute a \emph{connected} graph. 
Due to page limitations,  we omit the proofs for the bounds pertaining to $D$-connected networks: such bounds can be derived with techniques very similar to the ones used in this paper.

Table \ref{tab:consbounds} elucidates the relative advantages of different approaches to distributed computation. On static networks, the modified versions of the GHS  algorithm discussed in this paper \emph{simultaneously} achieve optimal time, message, byte, and storage complexity under mild assumptions regarding the consensus function, both for directional and broadcast communication, and both for sparse and dense graphs. On the other hand, the GHS algorithm is  not readily applicable to time-varying  networks and is sensitive to single-points of failure (since it relies on the construction of a spanning tree). In other words, a GHS algorithm has minimal robustness margins to the disruption of a communication channel.

A flooding algorithm is time-optimal, but as one can expect has poor message and byte  complexity. Also the storage complexity is worst among all considered algorithms. On the other hand, flooding is maximally robust to communication disruptions. Also, somewhat surprisingly, this study shows how a simple flooding algorithm outperforms an average-based algorithm (that solves the \emph{specific} consensus problem with consensus function $f(x) = \frac{1}{n}\mathbf{1}^T x$) with respect to all performance metrics except storage complexity (which is arguably a minor concern for modern embedded systems).

Finally, the hybrid clustering algorithm  introduced by the authors in \cite{FR-MP:13} has performance intermediate between those of flooding and GHS, as a function of a tuning parameter $m$.  The main advantage of this algorithm is to ``trade" some of the optimality of GHS with a tunable ``degree of robustness".

As discussed, GHS-like  algorithms (and also the hybrid algorithm in \cite{FR-MP:13}) can not be readily applied to dynamic settings. Note, however, 
 that the execution time of GHS  is $O(n)$, an order of magnitude faster than average-based algorithms: if reconfigurations are infrequent (i.e., their frequency is much lower than $O(1/n)$),  these algorithms  can indeed be applied to nominally dynamic networks. 
 
We conclude this paper with a discussion of the limitations of our analysis, which immediately reflect into a number of interesting directions for future research. First, optimal algorithms for $D$-connected networks are not currently known (of course, our lower bounds may not be tight), which represents a key area of study. Second, while locally-sensitive and extractive consensus functions represent a fairly large class of consensus functions, it is of interest to generalize our bounds on message, byte, and storage complexity even further. Third, we employ a worst-case approach over the classes of \emph{sparse} and \emph{dense} graphs. An interesting direction for future research would be (i) to derive bounds on a finer partition of the class of possible graphs, e.g. by parameterizing the graphs by their maximum node degree, and (ii) by embedding the problem within a probabilistic structure,  to derive performance bounds with respect to sets of graphs randomly drawn from a given probability distribution, so that one can derive measures of average (as opposed to worst-case) performance.  Fourth, our model essentially does not include the notion or robustness with respect to either stopping (i.e., for malfunctions) or byzantine (i.e., malicious) failures, which is an aspect of pivotal importance for the reliable deployment of cyber-physical systems. Finally, a practical implementations of the algorithms discussed in this paper could shed additional light on the relative benefits of the different approaches.
 
Overall, we hope that this work will prompt researchers in the field of multi-agent systems to compare their results against the fundamental lower bounds derived in this paper to properly evaluate the relative benefits of their approach.
\vspace{-0.0445cm}

\bibliographystyle{IEEEtran}
\bibliography{../../../bib/alias,../../../bib/main}

\begin{thebibliography}{10}
\providecommand{\url}[1]{#1}
\csname url@samestyle\endcsname
\providecommand{\newblock}{\relax}
\providecommand{\bibinfo}[2]{#2}
\providecommand{\BIBentrySTDinterwordspacing}{\spaceskip=0pt\relax}
\providecommand{\BIBentryALTinterwordstretchfactor}{4}
\providecommand{\BIBentryALTinterwordspacing}{\spaceskip=\fontdimen2\font plus
\BIBentryALTinterwordstretchfactor\fontdimen3\font minus
  \fontdimen4\font\relax}
\providecommand{\BIBforeignlanguage}[2]{{%
\expandafter\ifx\csname l@#1\endcsname\relax
\typeout{** WARNING: IEEEtran.bst: No hyphenation pattern has been}%
\typeout{** loaded for the language `#1'. Using the pattern for}%
\typeout{** the default language instead.}%
\else
\language=\csname l@#1\endcsname
\fi
#2}}
\providecommand{\BIBdecl}{\relax}
\BIBdecl

\bibitem{ROS:07}
R.~Olfati-Saber, ``Distributed {K}alman filtering for sensor networks,'' in
  \emph{Proc. {IEEE} Conf. on Decision and Control}, New Orleans, LA, Dec.
  2007, pp. 5492--5498.

\bibitem{WR-RWB-EMA:07}
W.~Ren, R.~W. Beard, and E.~M. Atkins, ``Information consensus in multivehicle
  cooperative control: {C}ollective group behavior through local interaction,''
  \emph{{IEEE} Control Systems Magazine}, vol.~27, no.~2, pp. 71--82, 2007.

\bibitem{MdW-BC:09}
M.~de~Weerdt and B.~Clement, ``Introduction to planning in multiagent
  systems,'' \emph{Multiagent and Grid Sys.}, vol.~5, no.~4, pp. 345--355,
  2009.

\bibitem{AJ-JL-ASM:02}
A.~Jadbabaie, J.~Lin, and A.~S. Morse, ``Coordination of groups of mobile
  autonomous agents using nearest neighbor rules,'' \emph{IEEE Transactions on
  Automatic Control}, vol.~48, no.~6, pp. 988--1001, 2003.

\bibitem{ROS-RMM:03c}
R.~Olfati-Saber and R.~M. Murray, ``Consensus problems in networks of agents
  with switching topology and time-delays,'' \emph{IEEE Transactions on
  Automatic Control}, vol.~49, no.~9, pp. 1520--1533, 2004.

\bibitem{AO:10}
A.~Olshevsky, ``Efficient information aggregation strategies for distributed
  control and signal processing,'' Ph.D. dissertation, MIT - EECS Department,
  September 2010.

\bibitem{NL:96}
N.~A. Lynch, \emph{Distributed Algorithms}.\hskip 1em plus 0.5em minus
  0.4em\relax San Francisco, CA, USA: Morgan Kaufmann Publishers Inc., 1996.

\bibitem{AO-JNT:07}
A.~Olshevsky and J.~N. Tsitsiklis, ``Convergence speed in distributed consensus
  and averaging,'' \emph{SIAM Journal on Control and Optimization}, vol.~48,
  no.~1, pp. 33--55, 2009.

\bibitem{MS-TS-TT:92}
M.~Suzuki, T.~Sasaki, and T.~Tsuchiya, ``Digital acoustic image transmission
  system for deep-sea research submersible,'' in \emph{OCEANS '92. Mastering
  the Oceans Through Technology. Proceedings.}, vol.~2, 1992, pp. 567--570.

\bibitem{DRY-AMB-BBW:91}
D.~R. Yoerger, A.~M. Bradley, and B.~B. Walden, ``The {Autonomous} {Benthic}
  {Explorer} ({ABE}): an {AUV} optimized for deep seafloor studies,'' in
  \emph{Proceedings of the seventh international symposium on unmanned
  untethered submersible technology (UUST91)}, 1991, pp. 60--70.

\bibitem{RGG-PAH-PMS:83}
R.~G. Gallager, P.~A. Humblet, and P.~M. Spira, ``A distributed algorithm for
  minimum-weight spanning trees,'' \emph{ACM Transactions on Programming
  Languages and Systems (TOPLAS)}, vol.~5, no.~1, pp. 66--77, 1983.

\bibitem{FR-MP:14b}
F.~Rossi and M.~Pavone, ``Distributed consensus with mixed time/communication
  bandwidth performance metrics,'' in \emph{Communication, Control, and
  Computing (Allerton), 2014 52st Annual Allerton Conference on}, 2014, in
  press.

\bibitem{BAw:87}
B.~Awerbuch, ``Optimal distributed algorithms for minimum weight spanning tree,
  counting, leader election, and related problems,'' in \emph{Proceedings of
  the nineteenth annual ACM symposium on Theory of computing}.\hskip 1em plus
  0.5em minus 0.4em\relax ACM, 1987, pp. 230--240.

\bibitem{FR-MP:13}
F.~Rossi and M.~Pavone, ``Decentralized decision-making on robotic networks
  with hybrid performance metrics,'' in \emph{Communication, Control, and
  Computing (Allerton), 2013 51st Annual Allerton Conference on}, Oct 2013, pp.
  358--365.

\bibitem{GF-NL:87}
G.~N. Frederickson and N.~A. Lynch, ``Electing a leader in a synchronous
  ring,'' \emph{Journal of the Association for Computing Machinery}, vol.~34,
  no.~1, pp. 98--115, 1987.

\bibitem{MS:85}
M.~Snir, ``On parallel searching,'' \emph{SIAM Journal on Computing}, vol.~14,
  no.~3, pp. 688--708, 1985.

\bibitem{EK-SM-SZ:84}
E.~Korach, S.~Moran, and S.~Zaks, ``Tight lower and upper bounds for some
  distributed algorithms for a complete network of processors,'' in
  \emph{Proceedings of the third annual ACM symposium on Principles of
  distributed computing}.\hskip 1em plus 0.5em minus 0.4em\relax ACM, 1984, pp.
  199--207.

\bibitem{FPR:30}
F.~P. Ramsey, ``On a problem of formal logic,'' \emph{Proceedings of the London
  Mathematical Society}, vol. s2-30, no.~1, pp. 264--286, 1930.

\end{thebibliography}
\end{document}